\documentclass[10pt,journal,final,finalsubmission,twocolumn]{IEEEtran}
\usepackage{bm}
\usepackage{mathrsfs}
\usepackage{textcomp}
\usepackage{epsfig}
\usepackage{indentfirst}
\usepackage{amsmath}
\usepackage{amssymb}
\usepackage{array}
\usepackage{multirow}
\usepackage{graphicx}
\usepackage{subfigure}

\usepackage{graphicx}
\usepackage{epstopdf}

\usepackage{amsthm}
\newtheorem{theorem}{Theorem}%[section]
\newtheorem{lemma}{Lemma}

\theoremstyle{definition}

\theoremstyle{remark}

\newcommand{\trace}{\mathop{\mathrm{tr}}}%\trace
\newcommand{\rank}{\mathop{\mathrm{rank}}}%\trace
\newcommand{\diag}{\mathop{\mathrm{diag}}}%\trace

\usepackage{cite}
\begin{document}
% %
% % paper title
% % can use linebreaks \\ within to get better formatting as desired
\title{{Energy-Throughput Tradeoff in Sustainable Cloud-RAN with Energy Harvesting}}
\author{{Zhao Chen, Ziru Chen, Lin X. Cai, and Yu Cheng}\\
Department of Electrical and Computer Engineering, Illinois Institute of Technology, Chicago, USA\\
Emails: zchen84@iit.edu, zchen71@hawk.iit.edu, \{lincai, cheng\}@iit.edu}

% % make the title area
\maketitle
\thispagestyle{empty}

% \linespread{1.6}

% \begin{abstract}% \end{abstract}

\begin{abstract}
In this paper, we investigate joint beamforming for energy-throughput tradeoff in a sustainable cloud radio access network system, where multiple base stations (BSs) powered by independent renewable energy sources will collaboratively transmit wireless information and energy to the data receiver and the energy receiver simultaneously.
In order to obtain the optimal joint beamforming design over a finite time horizon, we formulate an optimization problem to maximize the throughput of the data receiver while guaranteeing sufficient RF charged energy of the energy receiver.
Although such problem is non-convex, it can be relaxed into a convex form and upper bounded by the optimal value of the relaxed problem. We further prove tightness of the upper bound by showing the optimal solution to the relaxed problem is rank one.
Motivated by the optimal solution, an efficient online algorithm is also proposed for practical implementation. Finally, extensive simulations are performed to verify the superiority of the proposed joint beamforming strategy to other beamforming designs.
\end{abstract}


\section{Introduction}

As an emerging network architecture that incorporates cloud computing into wireless mobile networks, cloud radio access network (Cloud-RAN) has been proposed and will play a key role in the fifth generation (5G) communication system~\cite{chih2014toward}.
As shown in Fig. \ref{fig.arch}, all the base stations (BSs) in a Cloud-RAN system will be connected to a central processor (CP) for baseband processing via backhauls.
In order to mitigate the inter-cell interference and improve the overall throughput of the network, each user will receive distributed beamforming signals from a cluster of BSs in the downlink, while the wireless signals of each user will be received by a cluster of BSs in the uplink and then collaboratively processed at the CP. Note that both downlink and uplink beamforming designs have been widely addressed in~\cite{dai2014sparse,luo2015downlink,shi2014group}.

Energy harvesting (EH) from renewable energy sources\cite{piro2013hetnets} such as solar and wind powers, provides a green alternative for traditional on-grid power supplies in wireless communication systems.
Due to the intermittent nature of renewable sources, how to maximize system throughput by optimizing the consumption of randomly arriving energies has been extensively studied in the literature~\cite{ulukus2015energy}.
Particularly, radio-frequency (RF) energy of wireless signals can also be exploited to charge low-power devices remotely, which is referred to as RF charging.
For instance, sufficient energy can be collected by ambient signals from TV towers~\cite{vyas2013wehp} or Wi-Fi networks~\cite{ermeey2016indoor}, and dedicated energy signals have been considered to charge devices in the wireless powered communication network (WPCN)\cite{ju2014throughput}.
Moreover, since wireless signals can carry information and energy at the same time, it is possible to perform simultaneous wireless information and energy transfer (SWIPT)\cite{lu2015wireless} for either co-located information and energy receivers or separated located receivers.
Specifically, for multi-antenna systems, joint beamforming was explored to balance the RF charged energy and transmission rate between different receivers~\cite{xiang2012robust,xu2014multiuser,luo2015capacity}.
However, previous works did not consider the combination of renewable energy harvesting and RF charging to establish a fully sustainable wireless communication system, especially for the Cloud-RAN in future 5G systems.

\begin{figure}
\centering
\includegraphics[height = 6cm]{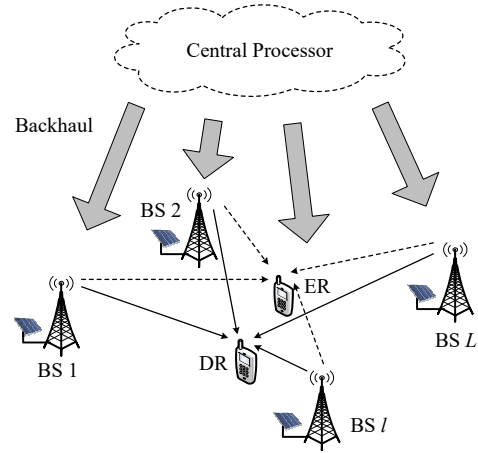}
\caption{A sustainable Cloud-RAN system, where renewable energy harvesting BSs transmit information and energy simultaneously to the data receiver (DR) and the energy receiver (ER) with joint beamforming.}
\label{fig.arch}
\vspace{-1em}
\end{figure}

In this paper, we consider a fully sustainable Cloud-RAN system, where all BSs are powered by independent renewable sources and broadcast to mobile users in the downlink simultaneously.
Each user can be either decode data or receive energy from wireless signals.
By providing the energy receiver with sufficient RF charged energy, the throughput of the data receiver will be maximized by joint energy and data beamforming design.
An offline optimization problem is formulated to study joint beamforming over a finite time horizon, which is originally non-convex and can be then solved optimally by relaxation.
Moreover, an efficient online algorithm is also proposed for implementation in practical systems.
By numerical evaluations, regions of energy-throughput tradeoff are built and performances of joint offline and online joint beamforming designed are extensively studied.

\section{System Model and Problem Formulation}\label{sec.I}

As shown in Fig. \ref{fig.arch}, we consider a downlink Cloud-RAN communication system, which consists of $L$ BSs, one data receiver and one energy receiver.
Powered by renewable energy sources such solar power and wind power, each BS is deployed in different locations and connected with the CP via a backhaul link.
In addition, it is assumed that each BS and each receiver are both equipped with one single antenna and operate on the same frequency band.
All the BSs will form a set $\mathcal{L} = \{1,\ldots,L\}$ and serve the mobile users cooperatively with wireless information and energy transmission by joint downlink beamforming.
Without loss of generality, we assume that there are totally $N$ time slots to be optimized. For each time slot $n \in \mathcal{N} = \{0,\ldots,N\}$, the amount of renewable energy harvested by the $l$th BS is denoted by $E_{l,n}$, which is assumed to be non-causally known for the offline problem.
We define $\boldsymbol{h} \in \mathbb{C}^{L \times 1}$ and $\boldsymbol{g} \in \mathbb{C}^{L \times 1}$ to be the quasi-static flat fading complex channel vectors from BS set $\mathcal{L}$ to the data receiver and the energy receiver, respectively, which is assumed to be constant throughout the short-term transmission and known at all BSs and users.

To build a sustainable Cloud-RAN system, we consider to maximize throughput of the data receiver while charging the energy receiver with sufficient RF energy. 
Thus, joint beamforming vector $\boldsymbol{w}_{n} = \left[{w}_{1,n},\ldots,{w}_{l,n},\ldots,{w}_{{L},n} \right]^T \in \mathbb{C}^{L \times 1}$ needs to be optimized for each time slot $n \in \mathcal{N}$, where ${w}_{l,n}$ is the beamforming weight of BS $l$.
The received signals of the data and energy receivers at slot $n$ can be written as
\begin{align}
y_{D,n} & = \boldsymbol{h}^H \boldsymbol{w}_{n} s_n + z_{D,n},\label{eq.rate_dr}  \\
y_{E,n} & = \boldsymbol{g}^H \boldsymbol{w}_{n} s_n + z_{E,n},\label{eq.charge_er}
\end{align}
where $s_{n} \sim \mathcal{CN}(0,1)$ denotes the symbol to the data receiver, and $z_{D,n}$ and $z_{E,n}$ are both additive white circularly symmetric complex Gaussian (CSCG) noises with variance $\sigma^2$.

Given the random energy profile $\{E_{l,n}\}_{n=1}^N$ of each BS $l \in \mathcal{L}$ and the energy receiver's RF charged energy constraint $q$, we can maximize throughput of the data receiver by optimizing the joint beamforming vector $\boldsymbol{w}_{n}$ as follows:
\begin{align}
\max_{\{\boldsymbol{w}_{n}\}_{n\in\mathcal{N}}} & T = \sum_{n=1}^N \log ( 1 + { |\boldsymbol{h}^H \boldsymbol{w}_{n}|^2}/{\sigma^2} ) \label{eq.gen_problem_sim}
 \\
\mathrm{s.t.} \hspace{1em} & Q = \sum_{n=1}^N \eta |\boldsymbol{g}^H \boldsymbol{w}_{n}|^2 \geq q, \label{eq.RF_energy_cons_sim} \\
& \sum_{t=1}^n \left| {w}_{l,t} \right|^2 \leq \sum_{t=1}^n E_{l,t}, \hspace{1em}\forall l \in \mathcal{L},\forall n \in \mathcal{N}, \label{eq.renew_energy_cons_sim}
\end{align}
where the RF charged energy of the energy receiver is constrained by \eqref{eq.RF_energy_cons_sim}, and \eqref{eq.renew_energy_cons_sim} guarantees the consumed energy of BS $l$ up to slot $n$ does not exceed the harvested energy.
Here, $\eta \in (0,1)$ is the energy conversion efficiency.
Note that problem \eqref{eq.gen_problem_sim} is a non-convex optimization problem due to the quadratic terms of $\boldsymbol{w}_{n}$ in the objective and constraint \eqref{eq.RF_energy_cons_sim}. Therefore, it is quite challenging to find the global optimum of the problem.

Notice that there exists a tradeoff between the data receiver's throughput $T$ and the energy receiver's RF charged energy $Q$, which characterizes boundary points of the achievable energy-throughput region
\begin{align}
\mathcal{C}_{E-T} = \{(Q,T):Q \leq q, T\leq f_{T}(q), q \leq q_{\max}\}, \nonumber
\end{align}
where $f_{T}(q)$ denotes the maximum throughput $T$ when given energy constraint $q$ in problem \eqref{eq.gen_problem_sim}, and $q_{\max}$ is the maximum achievable RF charged energy of the energy receiver.
In the next section, we will present how to solve problem \eqref{eq.gen_problem_sim} and derive the achievable energy-throughput region $\mathcal{C}_{E-T}$.

\section{Joint Beamforming Design for Energy-Throughput Tradeoff}

In this section, we will show the optimization of $\boldsymbol{w}_{n}$ for each time slot $n\in \mathcal{N}$. Firstly, we will examine the value of $q_{\max}$ to guarantee feasibility of problem \eqref{eq.gen_problem_sim}. Then, the region $\mathcal{C}_{E-T}$ can be derived by solving problem \eqref{eq.gen_problem_sim} for $q \leq q_{\max}$.

\subsection{Maximum Achievable RF Charged Energy}
Ignoring the data receiver, $q_{\max}$ can be obtained by maximizing RF charged energy of the energy receiver as follows
\begin{align}
\max_{\{\boldsymbol{w}_{n}\}_{n\in\mathcal{N}}} & \hspace{1em} Q
 \label{eq.q_max_pro} \\
\mathrm{s.t.} \hspace{1em} & \hspace{0.5em}\eqref{eq.renew_energy_cons_sim}. \nonumber
\end{align}
The traditional maximum ratio combining (MRC) strategy for multi-input single-output (MISO) with sum power constraint cannot be applied, since each BS in the Cloud-RAN system is constrained by independent renewable energy sources.
% Although the original form of problem \eqref{eq.q_max_pro} is a non-convex quadratically constrained quadratic program (QCQP) that is generally difficult to solve, it can be approximated and solved efficiently by semidefinite relaxation (SDR)~\cite{luo2010semidefinite}.
Nevertheless, to obtain a closed-form solution, we can rewrite the total RF charged energy of the energy receiver
\begin{align}
Q  = \eta \boldsymbol{g}^H \textbf{W} \boldsymbol{g} = \eta \sum_{i=1}^L \sum_{j=1}^L g_i^H g_j \textbf{W}_{ij} = \eta\trace{(\textbf{W}\textbf{G})}, \label{eq.q_entry}
\end{align}
where we define $\textbf{W} = \sum_{n=1}^N \boldsymbol{W}_n = \sum_{n=1}^N\boldsymbol{w}_{n}\boldsymbol{w}_{n}^H$. It can be easily seen that each entry of $\textbf{W}$ can be written by $\textbf{W}_{ij} = \sum_{n=1}^N {w}_{i,n} {w}_{j,n}^H$ for each $i,j \in \mathcal{L}$.

\begin{lemma}\label{lm.opt_energy_not_capacity}
The optimal solution to problem \eqref{eq.q_max_pro} can be denoted by $\boldsymbol{w}_{n}^*=\sqrt{P_n}\boldsymbol{w}_{0}^*e^{j\theta_n}$, as long as it satisfies the energy constraints $\sum_{t=1}^n P_t |{w}_{l,0}|^2 \leq \sum_{t=1}^n E_{l,t}$ for each BS $l \in \mathcal{L}$ at each time slot $n \in \mathcal{N}$.
Here, the optimal unit beamforming vector $\boldsymbol{w}_{0}^*$ is defined by
\begin{align}
\begin{small}
\boldsymbol{w}^*_{0} = \left(\frac{g_{1}}{|g_{1}|}\sqrt\frac{{\sum_{n=1}^N E_{1,n}}}{\sum_{l=1}^L\sum_{n=1}^N E_{l,n}},\ldots,\frac{g_{L}}{|g_{L}|}\sqrt\frac{{\sum_{n=1}^N E_{L,n}}}{\sum_{l=1}^L\sum_{n=1}^N E_{l,n}}\right)^T \nonumber
\end{small}
\end{align}
and $\theta_n \in [0,2\pi)$ is an arbitrary constant phase. Accordingly, the maximum achievable RF charged energy can be obtained by $q_{\max} = \eta\left(\sum_{l=1}^L {|g_{l}|} \sqrt{\sum_{n=1}^N E_{l,n}}\right)^2$.
\end{lemma}
\begin{proof}
Please refer to Appendix \ref{app.opt_energy_not_capacity}.
\end{proof}

Note that the optimal solution to problem \eqref{eq.q_max_pro} is not unique. For each time slot $n$, the optimal power consumption ratio of each BS $l\in\mathcal{L}$ is equal to its total energy harvesting ratio, i.e.
$\|{w}_{0,l}\|^2 = \frac{|w_{l,n}|^2}{\|\boldsymbol{w}_{n}\|^2} = \frac{\sum_{n=1}^N E_{l,n}}{\sum_{l=1}^L\sum_{n=1}^N E_{l,n}}$, which is referred to as \emph{optimal power ratio} of each BS $l$ for energy maximization.

% \begin{corollary}\label{cor.1}
% The maximum capacity when $q = q_{\max}$ can be given by
% \end{corollary}

\subsection{Achievable Energy-Throughput Region}

Obtaining $q_{\max}$, we can tackle problem \eqref{eq.gen_problem_sim} with any given $q \in [0,q_{\max}]$.
Following the same manipulation of $\{\boldsymbol{w}_{n}\}_{n\in\mathcal{N}}$ in \eqref{eq.q_entry}, problem \eqref{eq.gen_problem_sim} can be reformulated in terms of transmit covariance matrices $\{\boldsymbol{W}_{n}\}_{n\in\mathcal{N}}$ as follows,
\begin{align}
\hspace{-1em}\max_{\{\boldsymbol{W}_{n} \succeq 0\}_{n\in\mathcal{N}}} & \sum_{n=1}^N \log \left( 1 + { \boldsymbol{h}^H \boldsymbol{W}_{n}\boldsymbol{h}}/{\sigma^2} \right) \
 \label{eq.sim_pro_reform} \\
\mathrm{s.t.} \hspace{1.5em}
& \hspace{-0.5em}\trace{\left(\sum_{n=1}^N \boldsymbol{W}_n\textbf{G}\right)} \geq q/\eta, \label{eq.sim_pro_reform_q}\\
& \hspace{-0.5em} \trace{\left(\sum_{t=1}^n \boldsymbol{W}_{t} \boldsymbol{A}_{l}\right)} \leq \sum_{t=1}^n E_{l,t}, \hspace{0em} \forall l \in \mathcal{L},\forall n \in \mathcal{N}, \hspace{-1em}\label{eq.r_simp_linear_cons}
% & \left[\sum_{t=1}^n \boldsymbol{W}_{t} \right]_{ll} \leq \sum_{t=1}^n E_{l,t}, \hspace{1em}\forall l \in \mathcal{L},\forall n \in \mathcal{N},   \\
% & \boldsymbol{W}_n \succeq 0, \hspace{0.3em} \forall n \in \mathcal{N}, \nonumber
\end{align}
where the constraints in \eqref{eq.r_simp_linear_cons} follow from \eqref{eq.renew_energy_cons_sim} and the auxiliary matrices are defined as $\boldsymbol{A}_{l} = \diag{(\underbrace{0,\ldots,0}_{l-1},1,\underbrace{0,\ldots,0}_{L-l})}$.

Problem \eqref{eq.sim_pro_reform} is a convex optimization problem, which can be efficiently solved by the standard interior point algorithm~\cite{boyd2004convex}.
However, it still needs to be proved that each optimal $\boldsymbol{W}^*_{n}$ is rank one and then can be decomposed into the optimal vector $\boldsymbol{w}^*_{n}$ in problem \eqref{eq.gen_problem_sim}.
Otherwise problem \eqref{eq.sim_pro_reform} only gives the upper bound of problem \eqref{eq.gen_problem_sim}.
Later, we will show that the optimal matrix $\boldsymbol{W}^*_{n}$ satisfies $\rank{\left(\boldsymbol{W}^*_{n}\right)} \leq 1$ for each $n\in\mathcal{N}$.

For properties of the optimal solution to problem \eqref{eq.sim_pro_reform}, we study the following rate maximization problem with given transmission power $p_{l,n}$ for each BS $l$ and RF charging constraint $q_n$ for the energy receiver at slot $n$:
\begin{align}
\max_{\boldsymbol{W}_{n} \succeq 0}  \hspace{1em}  & c(\boldsymbol{W}_{n}) \label{eq.sub_problem} \\
\mathrm{s.t.} \hspace{1em}
& \trace{(\boldsymbol{W}_n\textbf{G})} \geq q_n , \nonumber\\
& \trace{\left( \boldsymbol{W}_{n} \boldsymbol{A}_{l}\right)} \leq p_{l,n}, \hspace{1em} \forall l \in \mathcal{L}, \nonumber
\end{align}
where we define $c(\boldsymbol{W}_{n}) = \log ( 1 + { \boldsymbol{h}^H \boldsymbol{W}_{n}\boldsymbol{h}}/{\sigma^2} )$. Letting $\boldsymbol{p}_n = \{p_{1,n},p_{2,n},\ldots,p_{L,n}\}$, the optimal value of problem \eqref{eq.sub_problem} can be denoted as function $f_R(\boldsymbol{p}_n,q_n)$, which is a concave function as proved in the following lemma.

\begin{lemma}\label{lm.concavity_r}
For each single time slot $n\in \mathcal{N}$, the maximum rate function $r_n = f_R(\boldsymbol{p}_n,q_n)$ is a concave function of $\left(\boldsymbol{p}_n,q_n\right)$.
\end{lemma}
\begin{proof}
Please refer to Appendix \ref{app.concavity_r}.
\end{proof}

Due to the concavity of power-rate capacity function, the direct water-filling (DWT) algorithm is well-known to maximize short-term system throughput~\cite{yang2012optimal,ozel2011transmission} for energy harvesting enabled point-to-point transmissions.
In the DWT algorithm, the optimal power allocation for each time slot will be non-decreasing and remain constant until the energy in the battery is exhausted, which motivates us to simplify problem \eqref{eq.sim_pro_reform}.

Similar to the DWT algorithm, we first process the energy profile $\{E_{l,n}\}_{n\in\mathcal{N}}$ of each individual BS $l\in \mathcal{L}$ and find its $k$th power changing slots as
\begin{align}
\vspace{-1em}
n_{l,k} = \arg \min_{n_{l,k} \leq N, n_{l,k} > n_{l,k-1}} \left\{\frac{\sum_{t = n_{l,k-1}+1}^{n_{l,k}} E_{l,t}}{n_{l,k} - n_{l,k-1}}\right\}, \label{eq.energy_changing_slots}
\end{align}
where $k \in \{1,\ldots,K_l\}$ with $K_l$ being the number of power changing slots, and $n_{l,0} = 0$, $n_{l,K_l} = N$.
Then, considering the independent random energy profiles for all the BSs, we combine all the energy changing slots obtained from \eqref{eq.energy_changing_slots} as an index set $\mathcal{N}_{M} = \bigcup_{l\in\mathcal{L}} \{n_{l,k_l}\}_{k_l=1}^{K_l}$, where duplicate changing slot indices will be all removed.
Finally, after sorting $\mathcal{N}_{M}$ in an ascending order by $n_1<n_2<\ldots<n_M$, all the $N$ time slots can be divided into $M = |\mathcal{N}_{M}|$ intervals, where each interval $m$ contains slots $\mathcal{I}_m = \{n_{m-1}+1,\ldots,n_{m}\}$ for any $m  \in \mathcal{M} =  \{1,\ldots,M\}$ and $n_0 = 0$. Besides, the length of the $m$th interval can be denoted by $I_m = n_{m}-n_{m-1}$.

Following the concavity of $f_R(\boldsymbol{p}_n,q_n)$, we can obtain some properties of the optimal solution to problem \eqref{eq.sim_pro_reform}.

\begin{lemma}\label{lm.equal_power_assignment}
In each interval $m \in \mathcal{M}$, the optimal solution to every time slot $n \in \mathcal{I}_m$ of problem \eqref{eq.sim_pro_reform} assigns the same power consumption and RF charging constraint for each BS and the energy receiver, i.e. $\left(\boldsymbol{p}^*_n,q^*_n\right) = \left(
\tilde{\boldsymbol{p}}_m,\tilde{q}_m\right)$.
\end{lemma}

\begin{proof}
For each interval $m$, suppose there are $S_m$ sub-intervals with different optimal  $\left(\tilde{\boldsymbol{p}}^{1}_m,\tilde{q}^{1}_m\right),\ldots,\left(
\tilde{\boldsymbol{p}}^{S_m}_m,\tilde{q}^{S_m}_m\right)$ of durations $I_m^{1},\ldots,I_{m}^{S_m}$.
Note that the energy harvesting constraints in \eqref{eq.r_simp_linear_cons} are still satisfied and the interval length $I_m = \sum_{s=1}^{S_m} I_{m}^s$.

Let $\left(\tilde{\boldsymbol{p}}_m,\tilde{q}_m\right) = \left(\sum_{s=1}^{S_m}\tilde{\boldsymbol{p}}_m^sI_{m}^s , \sum_{s=1}^{S_m}\tilde{{q}}_m^sI_{m}^s  \right)/{I_{m}}$.
Due to the concavity of $f_R(\boldsymbol{p}_n,q_n)$, we can derive from Lemma \ref{lm.concavity_r},
\begin{align}
\vspace{-1em}
f_R\left(\tilde{\boldsymbol{p}}_m,\tilde{q}_m\right)\cdot{I_m} \geq \sum_{s=1}^{S_m} f_R \left(
\tilde{\boldsymbol{p}}^{s}_m,\tilde{q}^{s}_m\right) \cdot I_{m}^s,
\end{align}
% \begin{align}
% &  f_R(\frac{I_m^{1}\tilde{\boldsymbol{p}}^{1}_m+\ldots+I_m^{1}\tilde{\boldsymbol{p}}^{M_m}_m}{I_m^{1} + \ldots + I_m^{M_m}}, \frac{I_m^{1}\tilde{{q}}^{1}_m+\ldots+I_m^{1}\tilde{{q}}^{M_m}_m}{I_m^{1} + \ldots + I_m^{M_m}}), \nonumber\\
% & \geq \frac{I_m^{1}f_R\left(\tilde{\boldsymbol{p}}^{1}_m,\tilde{q}^{1}_m\right) + \ldots + I_m^{M_m}, \nonumber
% \end{align}
% two different optimal power assignments $\left(\boldsymbol{p}^\prime_m,q^\prime_m\right)$ and $\left(\boldsymbol{p}^{\prime\prime}_m,q^{\prime\prime}_m\right)$ for two different subintervals $\mathcal{I}^{\prime}_m \subset \mathcal{I}_m$ and $\mathcal{I}^{\prime\prime}_m \subset \mathcal{I}_m$, respectively.
% For each BS $l \in \mathcal{L}$ with different power consumptions, i.e. $p_{l,m}^{\prime} \neq p_{l,m}^{\prime\prime}$, the average power consumption $\bar{p}_{l,m} = \frac{I^\prime_{m}p_{l,m}^{\prime} + I^{\prime\prime}_{m}p_{l,m}^{\prime\prime}}{I^\prime_{m} + I^{\prime\prime}_m}$.
which indicates another optimal strategy with equal $\left(\tilde{\boldsymbol{p}}_m,\tilde{q}_m\right)$ exists for each slot $n \in \mathcal{I}_m$ and achieves a higher throughput without violating the RF charging constraint \eqref{eq.sim_pro_reform_q}.
In addition, since there is no energy changing slot in interval $\mathcal{I}_m$, the optimal  $\left(\tilde{\boldsymbol{p}}_m,\tilde{q}_m\right)$ also satisfies the energy harvesting constraints by \eqref{eq.energy_changing_slots}. Thus, it contradicts with the assumption of different optimal power assignments, which completes the proof.
\end{proof}

From Lemma \ref{lm.equal_power_assignment}, it can be seen that for each slot $n \in \mathcal{I}_m$, the optimal joint beamforming vector $\boldsymbol{w}_n^* = \tilde{\boldsymbol{w}}_m$ and the optimal transmission rate of the data receiver $r_n^* = \tilde{r}_m = f_R\left(\tilde{\boldsymbol{p}}_m,\tilde{q}_m\right)$.
Thus, problem \eqref{eq.sim_pro_reform} can be simplified by optimizing for each interval $m \in \mathcal{M}$ instead of optimizing for each slot $n\in \mathcal{N}$,
\begin{align}
\vspace{-1em}
\hspace{-0.6em}\max_{\{\widetilde{\boldsymbol{W}}_{m} \succeq 0\}_{m\in\mathcal{M}}} & \sum_{m=1}^M I_{m}c(\widetilde{\boldsymbol{W}}_{m})  \label{eq.sim_pro_reform_int} \\
\mathrm{s.t.} \hspace{2.5em}
& \hspace{-1.5em} \trace{\left(\sum_{m=1}^M I_{m} \widetilde{\boldsymbol{W}}_m\textbf{G}\right)} \geq q/\eta, \nonumber\\
& \hspace{-1.5em} \trace{\left(\sum_{t =1}^m I_{t} \widetilde{\boldsymbol{W}}_{t} \boldsymbol{A}_{l}\right)} \leq \sum_{t^\prime =1}^{n_m} E_{l,t^\prime}, \hspace{0em} \forall l \in \mathcal{L},\forall m \in \mathcal{M}, \hspace{-0.2em} \nonumber
\end{align}
where $\widetilde{\boldsymbol{W}}_m = \tilde{\boldsymbol{w}}_m \tilde{\boldsymbol{w}}_m^H$ is the transmit beamforming covariance matrix for the interval $m$.
In order to derive a semi closed-form expression of the optimal solution to problem \eqref{eq.sim_pro_reform_int}, the Lagrangian function of problem \eqref{eq.sim_pro_reform_int} can be written by,
\begin{align}
&L(\{\widetilde{\boldsymbol{W}}_{m}\},\{\{\lambda_{l,m}\}\},\mu) \nonumber \\
% & =  \sum_{m}^M I_m c(\widetilde{\boldsymbol{W}}_{m}) + \mu \left(\trace{\left(\sum_{m=1}^M I_{m} \widetilde{\boldsymbol{W}}_m\textbf{G}\right)}  - q/\eta\right) \nonumber \\
% & \hspace{1em}- \sum_{l=1}^L \sum_{m=1}^M \lambda_{l,m} \left(\trace{\left(\sum_{t =1}^m I_{t} \widetilde{\boldsymbol{W}}_{t} \boldsymbol{A}_{l}\right)} - \sum_{t^\prime =1}^{n_m} E_{l,t^\prime}\right), \nonumber  \\
& = \sum_{m=1}^M I_m L_m + \sum_{l=1}^L \sum_{m=1}^M  \lambda_{l,m} \sum_{t^\prime =1}^{n_m} E_{l,t^\prime} - \mu q/\eta, \label{eq.original_langrangian}
\end{align}
where for the $m$th interval, we have
\begin{align}
L_m =  c(\widetilde{\boldsymbol{W}}_{m}) +\mu \trace{\left(\widetilde{\boldsymbol{W}}_m\textbf{G}\right)}-\sum_{l=1}^L \sum_{t=m}^M \lambda_{l,t} \trace{\left( \widetilde{\boldsymbol{W}}_{m} \boldsymbol{A}_{l}\right)}.  \nonumber
\end{align}
Then, from \eqref{eq.original_langrangian}, the dual function $g(\{\{\lambda_{l,m}\}\},\mu)$ can be defined as the maximum of the following optimization problem
\begin{align}
g(\{\{\lambda_{l,m}\}\},\mu) =\max_{\{\widetilde{\boldsymbol{W}}_{m} \succeq 0\}_{m\in\mathcal{M}}} &  L(\{\widetilde{\boldsymbol{W}}_{m}\},\{\{\lambda_{l,m}\}\},\mu). \nonumber
\end{align}
Since the original problem \eqref{eq.sim_pro_reform_int} is convex, the dual function reaches a minimum at the optimal value of the primal problem, i.e. the duality gap is zero.
Thus, the optimal value of problem \eqref{eq.sim_pro_reform_int} is equivalent to
$\min_{\{\{\lambda_{l,m} \geq 0\}\}, \mu \geq 0} g(\{\{\lambda_{l,m}\}\},\mu)$.
In the sequel, we will show how to compute the dual function and solve the minimization problem efficiently.

From \eqref{eq.original_langrangian}, the maximization of $L(\{\widetilde{\boldsymbol{W}}_{m}\},\{\{\lambda_{l,m}\}\},\mu)$ can be decomposed into sub-problems as follows,
\begin{align}
\max_{\widetilde{\boldsymbol{W}}_{m} \succeq 0}  \log \left( 1 + { \boldsymbol{h}^H \widetilde{\boldsymbol{W}}_{m}\boldsymbol{h}}/{\sigma^2} \right)
 - \trace{(\boldsymbol{B}_m \widetilde{\boldsymbol{W}}_{m})}, \label{eq.dual_decomp_sub}
\end{align}
where we define $\boldsymbol{B}_m = \diag\left(\sum_{t=m}^M \lambda_{1,t}, \ldots, \sum_{t=m}^M \lambda_{L,t}\right) -  \mu \textbf{G}$. It requires $\boldsymbol{B}_m$ to be positive definite, i.e. $\boldsymbol{B}_m \succ 0$. Otherwise the optimal value of problem \eqref{eq.dual_decomp_sub} will be unbounded above.
If the optimal dual solution to problem \eqref{eq.sim_pro_reform_int} are denoted as $\{\{\lambda^*_{l,m}\}_{l\in\mathcal{L}}\}_{m\in\mathcal{M}}$ and $\mu^*$, we have the following theorem.

\begin{theorem}\label{th.1}
The optimal solution to problem \eqref{eq.sim_pro_reform_int} is rank one for each interval $m\in\mathcal{M}$ and can be derived in the form of
\begin{align}
\widetilde{\boldsymbol{W}}_{m}^* = \boldsymbol{B}_m^{-\frac{1}{2}}\boldsymbol{v}_m \left(1 - \sigma^2/\tilde{h}_{m}\right)^+\boldsymbol{v}_m^H \boldsymbol{B}_m^{-\frac{1}{2}}, \nonumber
\end{align}
where we define $\boldsymbol{B}_m = \diag\left(\sum_{t=m}^M \lambda^*_{1,t}, \ldots, \sum_{t=m}^M \lambda^*_{L,t}\right) -  \mu^* \textbf{G}$, $\tilde{h}_{m} = \|\boldsymbol{h}^{H}\boldsymbol{B}_m^{-1/2}\|^2$ and $\boldsymbol{v}_m = \frac{\boldsymbol{h}^{H}\boldsymbol{B}_m^{-1/2}}{\|\boldsymbol{h}^{H}\boldsymbol{B}_m^{-1/2}\|}$.
\end{theorem}

\begin{proof}
Firstly, we introduce $\widehat{\boldsymbol{W}}_{m} = \boldsymbol{B}_m^{1/2}\widetilde{\boldsymbol{W}}_{m}\boldsymbol{B}_m^{1/2}$ as auxiliary matrices and reformulate sub-problem \eqref{eq.dual_decomp_sub} by
\begin{align}
\max_{\widehat{\boldsymbol{W}}_{m} \succeq 0}  \log \left( 1 + \frac{ \boldsymbol{h}\boldsymbol{B}_m^{-1/2}\widehat{\boldsymbol{W}}_{m}  \boldsymbol{B}_m^{-1/2}\boldsymbol{h}^H}{\sigma^2} \right)
 - \trace{(\widehat{\boldsymbol{W}}_{m} )}, \nonumber
\end{align}
which is equivalent to the problem of MIMO channel capacity maximization with sum-power constraint and can be solved by the standard water-filling algorithm for an equivalent Gaussian vector channel $\boldsymbol{h}^{H}\boldsymbol{B}_m^{-1/2}$~\cite{cover2012elements}.
According to the singular value decomposition (SVD) of the equivalent channel vector $\boldsymbol{h}^{H}\boldsymbol{B}_m^{-1/2} = 1\times \sqrt{\tilde{h}_{m}} \times \boldsymbol{v}_m^H$ with $\tilde{h}_{m} = \|\boldsymbol{h}^{H}\boldsymbol{B}_m^{-1/2}\|^2$ and $\boldsymbol{v}_m = \frac{\boldsymbol{h}^{H}\boldsymbol{B}_m^{-1/2}}{\|\boldsymbol{h}^{H}\boldsymbol{B}_m^{-1/2}\|}$, we can obtain the optimal solution as \begin{align}
\widehat{\boldsymbol{W}}_{m}^* = \boldsymbol{v}_m \left(1 - \sigma^2/\tilde{h}_{m}\right)^+\boldsymbol{v}_m^H,  \nonumber
\end{align}
from which we obtain the optimal solution $\widetilde{\boldsymbol{W}}_{m}^*$.
\end{proof}

As a result, for any time slot $n \in \mathcal{I}_m$, the optimal variance matrix ${\boldsymbol{W}}^*_{n}$ is rank one, which indicates that the maximum of problem \eqref{eq.sim_pro_reform_int} can be achieved by problem \eqref{eq.gen_problem_sim}. Thus, the optimal beamforming vectors to problem \eqref{eq.gen_problem_sim} can be given by
\begin{align}
{\boldsymbol{w}}^*_n = \tilde{\boldsymbol{w}}^*_m = \sqrt{\left(1 - \sigma^2/\tilde{h}_{m}\right)^+} \boldsymbol{B}_m^{-\frac{1}{2}}\boldsymbol{v}_m,
\end{align}
and the maximum achievable rate of the data receiver is
\begin{align}
r_n^* = \tilde{r}^*_m & = \log\left(1+\tilde{h}_{m}\left(1 - \sigma^2/\tilde{h}_{m}\right)^+/\sigma^2\right), \\
& = \left(2 \log\left(\|\boldsymbol{h}^{H}\boldsymbol{B}_m^{-1/2}\|/\sigma\right)\right)^+. \nonumber
\end{align}

\vspace{-1em}
\subsection{Online Algorithm}
In this part, we consider only casual information of the random energy arrivals is available at each BS. That is, joint beamforming vector ${\boldsymbol{w}}_n$ for each slot $n \in \mathcal{N}$ will be decided only by the energy arrivals of current and past slots, i.e. $\{E_{l,t}\}_{t\leq n}$, and the average energy harvesting rate $P_{H,l}$.
Motivated by Lemma \ref{lm.opt_energy_not_capacity}, to guarantee RF charged energy for the energy receiver, the online algorithm will assign power consumptions for all BSs proportional to their average energy harvesting rates, i.e. $\boldsymbol{p}_n = k_n\cdot(P_{H,1},\ldots,P_{H,L})$, where factor $k_n = \min_{l\in\mathcal{L}}\{\frac{b_{l,n}}{P_{H,l}}\}$ with $b_{l,n} = E_{l,n} + \sum_{t=1}^{n-1} \left(E_{l,t} - p_{l,t}\right)$ being the residual energy of BS $l \in \mathcal{L}$.
Then, the RF charging constraints for the energy receiver will be
\begin{align} \nonumber
q_{n} =
\begin{cases}
        \min\{\frac{k_n}{N}q,\hspace{0.3em}f_E(\boldsymbol{p}_n)\}, & \hbox{$n < N$;} \\
        q - \sum_{t=1}^{N-1} q_t, & \hbox{$n = N$;}
\end{cases}
\end{align}
where $\frac{k_n}{N}q$ is the expected RF charged energy when transmitting with factor $k_n$, and $f_E(\boldsymbol{p}_n) = \big(\sum_{l=1}^L|g_{l}|\sqrt{p_{l,n}}\big)^2$ is the maximum RF charged energy when transmitting with $\boldsymbol{p}_n$.
Therefore, for each time slot $n \in \mathcal{N}$, the optimal beamforming vector of the online algorithm can be obtained by solving problem \eqref{eq.sub_problem} with given $\boldsymbol{p}_n$ and $q_n$, and thus the maximum transmission rate will be $r_n = f_R(\boldsymbol{p}_n,q_n)$.

\section{Numerical Evaluations}

In this section, we will validate our results with numerical examples, where the proposed joint beamforming strategy will be compared with some other baseline schemes.
We consider a Cloud-RAN system with three BSs, one data receiver and one energy receiver, where the distances between every two BSs are $50$ meters and the distances from the data receiver and the energy receiver to the three BSs are $(29,29,29)$ meters and $(10,40,45.8)$ meters, respectively.
The average channel power gain is modeled as $10^{-3}a/{d^{\alpha}}$, where $d$ is the distance, $\alpha$ is the pathloss exponent set as $\alpha = 2.5$ and $a \sim \exp(1)$ is the Rayleigh fading.
Assume system bandwidth is $1$ MHz and the additive white Gaussian noise at the data receiver has a power spectral density $N_0 = 10^{-15}$ W/Hz.
We consider the number of time slots $N = 60$ with slot length $1$s.
For each BS $l\in\mathcal{L}$, the random energy arrivals follow a Poisson distribution with mean $p_{H,l} = 0.1 \mathrm{W}$.
For the energy receiver, the energy conversion efficiency factor is $\eta = 80\%$.

For the purpose of exposition, average throughput of the data receiver and average RF charging rate of the energy receiver are defined as $\bar{r} = \frac{T}{N}$ and $\bar{q} = \frac{Q}{N}$, respectively.
In addition, the channel correlation factor between the two receivers is defined as $\rho = \frac{|\boldsymbol{g}^{H}\boldsymbol{h}|}{\|\boldsymbol{g}\|\|\boldsymbol{h}\|}$.
In Fig. \ref{fig.region}, energy-throughput regions are shown for correlation factors of $\rho = 0.10$, $0.31$ and $0.51$.
It can be seen that the average throughput $\bar{r}$ will decrease as the average RF charging rate $\bar{q}$ increases for a fixed $\rho$.
Moreover, the achievable energy-throughput region will expand as the channel correlation increases, which reveals that the energy receiver can benefit more from highly correlated channels.

\begin{figure}
\centering
\includegraphics[height = 6cm]{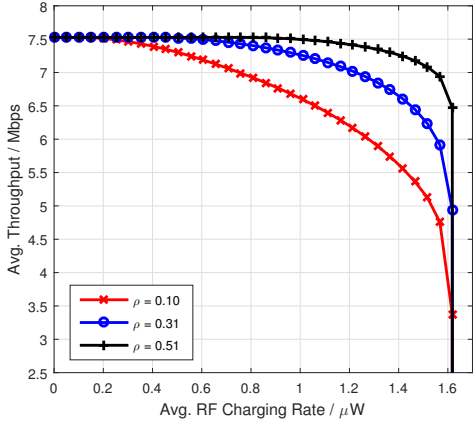}
\caption{Energy-throughput regions for different channel correlation factors.}
\label{fig.region}
\vspace{-1em}
\end{figure}
\begin{figure}
\centering
\includegraphics[height = 6cm]{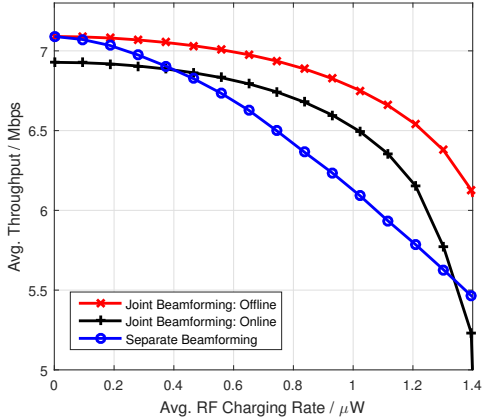}
\vspace{-1em}
\caption{Comparison of the average throughput versus average RF charging rate for different beamforming schemes.}
\label{fig.average_q}
\vspace{-1em}
\end{figure}

Now we evaluate the average performance of the proposed offline and online joint beamforming scheme over $100$ times random channel realizations.
A baseline scheme named separate beamforming is introduced, where energy beamforming is optimized for the first $N_E$ slots to guarantee the RF charging constraint, and then beamforming for data transmission is considered for the last $N-N_E$ slots.
% Note that different channel generations will result in different channel correlation factors.
In Fig. \ref{fig.average_q}, the average throughput $\bar{r}$ of all the three schemes decrease as the RF charging rate $\bar{q}$ increases. Moreover, the offline scheme outperforms the separate beamforming scheme and the performance gap will first increase and then stay roughly the same.
Meanwhile, the performance gap between the online and offline schemes is almost the same except for large RF charging constraint, which is due to the fact that the $q_{\max}$ of the offline scheme is larger than that of the online scheme.

\begin{figure}
\centering
\includegraphics[height = 7cm]{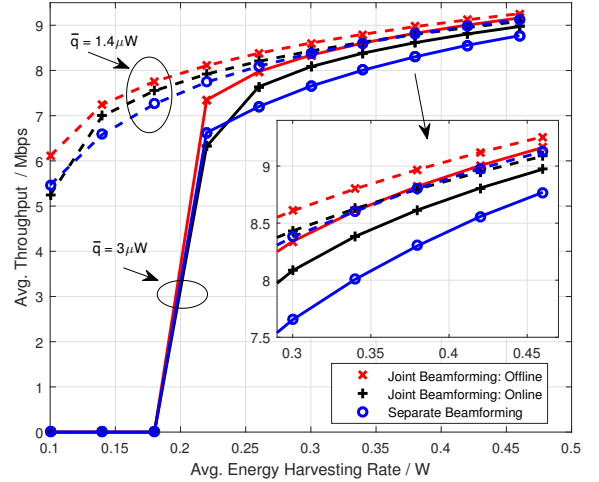}
\caption{Comparison of the average throughput versus average energy harvesting rate for different beamforming schemes.}
\label{fig.average_ph}
\vspace{-1em}
\end{figure}

In Fig. \ref{fig.average_ph}, the average throughput $\bar{r}$ versus average energy harvesting rate $P_{H,l}$ for different beamforming schemes are compared.
Here, $P_{H,l}$ are equal for all BSs.
On one hand, for a lower RF charging rate $\bar{q} = 1.4\mu$W, it can been seen that the gap between the offline scheme and the separate beamforming scheme decreases as $P_{H,l}$ increases, while the performance of the online and the separate beamforming schemes will be nearly the same for large $P_{H,l}$. On the other hand, for a larger RF charging rate $\bar{q} = 3\mu$W, both the offline and online schemes outperform the separate beamforming scheme, where the gaps become larger and decrease very slow as $P_{H,l}$ grows.

\section{Conclusion}

In this paper, joint energy and data beamforming design for energy-throughout tradeoff between one data receiver and one energy receiver has been investigated in a sustainable Cloud-RAN system.
An optimization problem has been formulated to maximize throughput of the data receiver while guaranteeing RF charged energy of the energy receiver. 
Both offline and online joint beamforming designs have been proposed and compared with a low-complexity baseline beamforming design by numerical simulations, which demonstrate the superiority of the proposed joint beamforming strategy.
Such joint beamforming strategy can be extended to support multiple data and energy receivers or massive MIMO BSs in our future work.

\appendix

\subsection{Proof of Lemma \ref{lm.opt_energy_not_capacity}}\label{app.opt_energy_not_capacity}
\begin{proof}
For $\forall i,j \in \mathcal{L}$, each entry $\textbf{W}_{ij}$ in the matrix $\textbf{W}$ satisfies,

\vspace{-1em}
\begin{small}
\begin{align}
\hspace{0.8em}& \hspace{-0.5em}|\textbf{W}_{ij}|^2 = \sum_{n=1}^N |{w}_{i,n}{w}_{j,n}^H|^2\!+\!\sum_{n_1=1}^N\!\sum_{\substack{n_2=1\\n_2 \neq n_1}}^N\!{w}_{i,n_1}{w}_{j,n_1}^H {w}_{i,n_2}^H{w}_{j,n_2},\nonumber \\
& \hspace{-1em}\leq \sum_{n=1}^N |{w}_{i,n}{w}_{j,n}^H|^2\!+\!\sum_{n_1=1}^N\!\sum_{\substack{n_2=1\\n_2 \neq n_1}}^N \hspace{-0.5em}\frac{|{w}_{i,n_1}{w}_{j,n_2}|^2 + |{w}_{i,n_2}^H{w}_{j,n_1}^H|^2}{2},\hspace{-0.5em}\nonumber\\
% & \hspace{-1em}= \sum_{n_1=1}^N \sum_{n_2=1}^N |{w}_{i,n_1}{w}_{j,n_2}|^2, \nonumber\\
& \hspace{-1em}= \sum_{n_1=1}^N |{w}_{i,n_1}|^2 \sum_{n_2=1}^N |{w}_{j,n_2}|^2, \nonumber
\end{align}
\vspace{-1em}
\end{small}

\hspace{-1em}where the equality holds if and only if for $\forall n_1,n_2 \in \mathcal{N}$,

\vspace{-1em}
\begin{small}
\begin{align}
{w}_{i,n_1}{w}_{j,n_2} =\left({w}_{i,n_2}^H{w}_{j,n_1}^H\right)^H \hspace{0.5em}\Leftrightarrow \hspace{0.5em}\frac{{w}_{i,n_1}}{{w}_{j,n_1}}=\frac{{w}_{i,n_2}}{{w}_{j,n_2}}. \nonumber
\end{align}
\vspace{-1em}
\end{small}

\hspace{-1em}It indicates that for each slot $n \in \mathcal{N}$, the beamforming weight $w_{l,n}$ of each BS $l \in \mathcal{L}$ is proportional to that of other BSs, i.e. $w_{l,n} = \sqrt{P_n} w_{l,0}e^{j\theta_n}$, and thus we have $\boldsymbol{w}_{n} =\sqrt{P_n}\boldsymbol{w}_{0}e^{j\theta_n}$, where $\theta_n \in [0,2\pi)$ is an arbitrary constant phase, and $P_n$ is the sum power of all BSs and $\boldsymbol{w}_{0} = (w_{1,0},\ldots,w_{L,0}) \in \mathbb{C}^{1\times L}$ is a unit vector that will be determined later.
Then, we can rewrite $\textbf{W}_{ij} = \sum_{n=1}^N P_n {w}_{i,0}{w}_{j,0}^H$ and obtain by \eqref{eq.q_entry} as follows,

\vspace{-1em}
\begin{small}
\begin{align}
Q & = \eta\sum_{i=1}^L \sum_{j=1}^L g_i^H g_j \sum_{n=1}^N P_n {w}_{i,0}{w}_{j,0}^H, \nonumber\\
& \leq \eta \sum_{n=1}^N P_n \sum_{i=1}^L \sum_{j=1}^L |g_i| |g_j||{w}_{i,0}||{w}_{j,0}|, \label{eq.abs_ineq} \\
& = \eta \sum_{n=1}^N P_n \bigg(\sum_{l=1}^L |g_l||w_{l,0}|\bigg)^2, \nonumber \\
& \leq \eta \sum_{l=1}^L \sum_{n=1}^N E_{l,n} \bigg(\sum_{l=1}^L |g_l||w_{l,0}|\bigg)^2, \label{eq.total_e_ineq}
\end{align}
\vspace{-1em}
\end{small}

\hspace{-1em}where the equality in \eqref{eq.abs_ineq} holds if and only if each complex number $g_i^H g_j{w}_{i,0}{w}_{j,0}^H$ has the same phase.
Thus, the phase of complex numbers $w_{l,0}$ and $g_l$ should be equal, i.e. $\frac{w_{l,0}}{|w_{l,0}|}\!=\!\frac{g_l}{|g_l|}$.
Besides, the equality $\sum_{n=1}^N P_n = \sum_{l=1}^L \sum_{n=1}^N E_{l,n}$ in \eqref{eq.total_e_ineq} holds if and only if for $\forall l \in \mathcal{L}$, $\sum_{t=1}^n P_t |w^*_{l,0}|^2 \leq \sum_{t=1}^n E_{l,t}$ for $\forall n \in \mathcal{N}$ and $|{w}^*_{l,0}| = \sqrt{\frac{\sum_{n=1}^N E_{l,n}}{\sum_{l=1}^L\sum_{n=1}^N E_{l,n}}}$, where the first condition is due to the energy constraints in \eqref{eq.renew_energy_cons_sim}, and the second condition can be proved by contradiction.
Suppose that any $w_{l,0}^\prime \neq w_{l,0}^*$, we have $\sum_{t=1}^N |w_{l,t}|^2 = \sum_{t=1}^N P_t |w_{l,0}|^2 \neq \sum_{t=1}^N E_{l,t}$, which contradicts with the equality in \eqref{eq.total_e_ineq}.
Thus, we obtain the optimal $\boldsymbol{w}^*_{0}$ and $q_{\max}$ in Lemma \ref{lm.opt_energy_not_capacity}.
\end{proof}

\subsection{Proof of Lemma \ref{lm.concavity_r}}\label{app.concavity_r}
\begin{proof}
Since problem \eqref{eq.sub_problem} is a convex optimization problem, it has a zero duality gap and thus can be solved by the Lagrange duality method

\vspace{-1em}
\begin{small}
\begin{align}
f_R(\boldsymbol{p}_n,q_n) & = \min_{\{\lambda_{l,n}\}_{l\in \mathcal{L}}, \mu} \max_{\boldsymbol{W}_{n} \succeq 0} c(\boldsymbol{W}_{n}) + \mu\left(\trace{(\boldsymbol{W}_n \textbf{G})} - q_n\right) \nonumber\\
& - \sum_{l = 1}^L \lambda_{l,n} \left(\trace{\left( \boldsymbol{W}_{n} \boldsymbol{A}_{l}\right)} - p_{l,n}\right),
\end{align}
\vspace{-1em}
\end{small}

\hspace{-1em}where $\{\lambda_{l,n}\}_{l\in \mathcal{L}}$ and $\mu$ are non-negative dual variables.
For a given $\left(\boldsymbol{p}_n,q_n\right)$, let $\boldsymbol{W}_{n}^*$, $\{\lambda^*_{l,n}\}_{l\in \mathcal{L}}$ and $\mu^*$ denote the optimal primal and dual variables of problem \eqref{eq.sub_problem}.
Suppose $\left(\boldsymbol{p}_n,q_n\right) = \alpha \left(\hat{\boldsymbol{p}}_n,\hat{q}_n\right) + (1-\alpha)\left(\check{\boldsymbol{p}}_n,\check{q}_n\right)$ with any $\alpha \in [0,1]$, we have

\begin{small}
\vspace{-1em}
\begin{align}
&\hspace{0.2em} f_R(\hat{\boldsymbol{p}}_n,\hat{q}_n) \nonumber \\
&\hspace{-0.2em} = c(\hat{\boldsymbol{W}}_{n}^*)\!+\!\hat{\mu}^*(\trace{(\hat{\boldsymbol{W}}_{n}^* \textbf{G})}\!-\!\hat{q}_n)\!-\!\sum_{l = 1}^L\!\hat{\lambda}^*_{l,n} (\trace{( \hat{\boldsymbol{W}}_{n}^* \boldsymbol{A}_{l})}\!-\!\hat{p}_{l,n}), \nonumber \\
&\hspace{-0.2em} \overset{(a)}{\leq}\!c(\hat{\boldsymbol{W}}_{n}^*)\!+\!{\mu}^*(\trace{(\hat{\boldsymbol{W}}_{n}^* \textbf{G})}\!-\!\hat{q}_n)\!-\!\sum_{l = 1}^L\!{\lambda}^*_{l,n} (\trace{( \hat{\boldsymbol{W}}_{n}^* \boldsymbol{A}_{l})}\!-\!\hat{p}_{l,n}), \nonumber \\
&\hspace{-0.2em} \overset{(b)}{\leq}\!c({\boldsymbol{W}}_{n}^*)\!+\!{\mu}^*(\trace{({\boldsymbol{W}}_{n}^* \textbf{G})}\!-\!\hat{q}_n)\!-\!\sum_{l = 1}^L\!{\lambda}^*_{l,n} (\trace{( {\boldsymbol{W}}_{n}^* \boldsymbol{A}_{l})}\!-\!\hat{p}_{l,n}), \nonumber
\end{align}
\vspace{-1em}
\end{small}

\hspace{-1em}where (a) holds because the dual variables $\{\hat{\lambda}^*_{l,n}\}_{l\in \mathcal{L}}$ and $\hat{\mu}^*$ minimize the dual function of $f_R(\hat{\boldsymbol{p}}_n,\hat{q}_n)$, and (b) holds since primal variable ${\boldsymbol{W}}_{n}^*$ maximizes the Lagrangian function of $f_R(\boldsymbol{p}_n,q_n)$ with given dual variables $\{\lambda^*_{l,n}\}_{l\in \mathcal{L}}$ and $\mu^*$.
Similarly, we can obtain the inequality for $f_R(\check{\boldsymbol{p}}_n,\check{q}_n)$. 

Therefore, with the obtained inequalities for $f_R(\hat{\boldsymbol{p}}_n,\hat{q}_n)$ and $f_R(\check{\boldsymbol{p}}_n,\check{q}_n)$, we have the following relationship

\begin{small}
\vspace{-1em}
\begin{align}
f_R(\boldsymbol{p}_n,q_n) \geq \alpha f_R(\hat{\boldsymbol{p}}_n,\hat{q}_n) + (1-\alpha) f_R(\check{\boldsymbol{p}}_n,\check{q}_n),
\end{align}
\vspace{-1em}
\end{small}

\hspace{-1em}where the inequality holds for $\alpha \in (0,1)$ and $f_R(\boldsymbol{p}_n,q_n) \neq 0$. Note that $f_R(\boldsymbol{p}_n,q_n) = 0$ if and only if $q_n \geq f_E(\boldsymbol{p}_n)$, i.e. the maximum RF charged energy of the energy receiver given power consumption $\boldsymbol{p}_n$.
Therefore, $f_R(\boldsymbol{p}_n,q_n)$ is a strictly concave function for $q_n \in [0, f_E(\boldsymbol{p}_n))$.
\end{proof}

\vspace{0.1em}
\section*{Acknowledgement}
This work is supported in part by the NSF under grants CNS-1053777, CNS-1320736, ECCS-1610874, and Career award ECCS-1554576. The work of Yu Cheng is also supported in part by National Natural Science Foundation of China under grant 61628107.

\end{document}